\documentclass[a4paper,12pt]{article}
\usepackage{amssymb,amsmath,amsthm}
\usepackage{graphicx}

\usepackage{amsfonts}
\usepackage{latexsym}
\usepackage{color}
\usepackage[english]{babel}
\usepackage{hyperref}
\usepackage{enumitem}

\numberwithin{equation}{section}

\newtheorem{lemma}{Lemma}[section]
\newtheorem{theorem}[lemma]{Theorem}
\newtheorem{proposition}[lemma]{Proposition}
\newtheorem{corollary}[lemma]{Corollary}
\newtheorem{remark}[lemma]{Remark}

\newtheorem{definition}[lemma]{Definition}

\def\dm{(-\Delta)^{M/2}}
\def\hzero{h^0}
\def\solh{g}
\def\Solh{G}

\def\toro{{\T}}
\def\im{{\rm i}}
\def\es{{\rm e}}
\def\op#1{Op^w(#1)}
\def\fke{{\frak e}}
\def\cV{{\mathcal{V}}}
\def\cL{H}

\def\cO{{\mathcal{O}}}
\def\cA{{\mathcal{A}}}
\def\cU{{\mathcal{U}}}

\def\ops#1{{OPS}^{#1}_{\delta}}
\def\Omegar{\Omega^{(R)}}
\def\aar{\Omega^{(R,c)}}
\def\aau{\widetilde{A}_k}
\def\moyal#1#2{\left\{#1;#2\right\}_M}
\def\poisson#1#2{\left\{#1;#2\right\}}
\def\sm#1{S^{#1}_{\delta}}
\def\adm#1{ad^M_{#1}}

\newcommand{\R}{\mathbb R}
\newcommand{\C}{\mathbb C}

\newcommand{\Z}{\mathbb Z}

\newcommand{\N}{\mathbb N}
\newcommand{\T}{\mathbb T}

\newcommand{\ii }{{\rm i} }

\newcommand{\vphi}{\varphi }

\newcommand{\csi}{\xi}

\begin{document}


\title{{\bf  On the spectrum of the
    Schr\"odinger operator on $\T^d$: a normal form approach }}

\date{}


\author{ Dario Bambusi\footnote{Dipartimento di Matematica, Universit\`a degli Studi di Milano, Via Saldini 50, I-20133
Milano. 
 \textit{Email: } \texttt{dario.bambusi@unimi.it}}, Beatrice Langella\footnote{Dipartimento di Matematica, Universit\`a degli Studi di Milano, Via Saldini 50, I-20133
Milano.
 \textit{Email: } \texttt{beatrice.langella@unimi.it}}, Riccardo Montalto \footnote{Dipartimento di Matematica, Universit\`a degli Studi di Milano, Via Saldini 50, I-20133
Milano.
 \textit{Email: } \texttt{riccardo.montalto@unimi.it}}
}

\maketitle

\begin{abstract}
  In this paper we study the spectrum of the operator
  \begin{equation}
  \label{ope}
\cL:=\dm+\cV\ , \quad M>0\ ,
\end{equation}
on $L^2(\R^d/\Gamma)$, with $\Gamma$ a maximal dimension lattice in
$\R^d$ and $\cV$ a pseudodifferential operator of order strictly
smaller than $M$. We prove that most of its eigenvalues admit the asymptotic
expansion 
\begin{equation}
  \label{sim}
\lambda_\xi=|\xi|^M+Z(\xi)+O(\left|\xi\right|^{-\infty})\ ,
\end{equation}
where $Z$ is a $C^\infty(\R^d)$ function (symbol) and $\xi\in\Gamma^*$
(the dual lattice of $\Gamma$).
\end{abstract}
\noindent


\medskip



\section{Introduction}\label{intro}

Let $\Gamma$ be a 
lattice of dimension $d$ in $\R^d$, with basis ${\bf e}_1, {\bf e}_2,
\ldots, {\bf e}_d$, namely 
\begin{equation}\label{definizione Gamma}
\Gamma := \Big\{  \sum_{i = 1}^d k_i {\bf e}_i : k_1, \ldots, k_d \in
\Z\Big\}\ ,
\end{equation}
and define 
\begin{equation}
  \label{toro}
 \T^d_\Gamma := \R^d / \Gamma\,.  
\end{equation}
{In this paper we study the asymptotic behavior of a large part of the
spectrum of the operator \eqref{ope} in $L^2(\T^d_\Gamma )$ and we
prove that there exists a smooth function (actually a symbol admitting
a full asymptotic expansion) $Z(\xi)$ s.t. for most points
$\xi\in\Gamma^*$ the corresponding eigenvalue $\lambda_\xi$ is given
by \eqref{sim}.  By most points we intend that they form a set of
density one (for a more precise estimate see eq. \eqref{density}
below). We remark that a particular case which is included in our
theory is that of the classical Sturm Liouville operators
\begin{equation}
  \label{sturm}
-\Delta+V(x) \ ,
\end{equation}
with
periodic or Floquet boundary conditions.}


{The spectrum of the Sturm Liouville operator
\eqref{sturm} was studied in \cite{FKT90,Fri90} whose results are in
the same spirit of our ones. We recall that in \cite{FKT90} it was
proven that if $V$ is a sufficiently smooth potential with zero
average and the
lattice is generic, then, for $\xi$
in a set of density 1, there are two eigenvalues
$\lambda_{\pm\xi}$ in the}
interval
$$
\left[|\xi|^2-\frac{1}{|\xi|^{2-\epsilon}},|\xi|^2+\frac{1}{|\xi|^{2-\epsilon}}\right]\ .
$$
In \cite{FKT90} the result was proven for $d=2,3$, while in \cite{Fri90}
the result was extended to general dimension (for further results see
{\cite{vel15,Wang11})}. 

Karpeshina \cite{Kar96,Kar97} studied in detail the case of Floquet boundary
conditions for bounded perturbations. She proved that, generically (in
the Floquet parameters), most eigenvalues are simple and she gave a
full asymptotic expansion of each one of these simple eigenvalues.

{In the present paper we get a full asymptotic expansion of the
eigenvalues studied in \cite{FKT90,Fri90}, getting in particular that,
for $C^{\infty}$ potential, one has} 
\begin{equation}
\label{split.eq}
\lambda_\xi-\lambda_{-\xi}=O(|\xi|^{-\infty})\ .
\end{equation}
Such a property is well known in dimension 1, but, as far as we
know, was not known in higher dimensions.  Furthermore the main
improvement that we get here is that we are able to deal with unbounded
perturbations. 

The proof is a development of the procedure already used
in \cite{1,2,3} (see also \cite{BBM14,BGMR1}) and consists of a
normal form procedure allowing to conjugate the operator \eqref{ope}
to a Fourier multiplier plus an operator which is smoothing at all
orders. This is constructed by quantizing the classical normal form
procedure applied to the symbol of the operator $\cL$. The main
difference with respect to \cite{1,2,3} is that in our problem the
resonances of the classical Hamiltonian system corresponding to the
unperturbed operator ($-\Delta$, in our case) depend on the point of
the phase space, so we restrict our construction to the
nonresonant regions of the phase space. {This is the reason why we only get
  the result for most eigenvalues}. A detailed heuristic
description of the proof is given in Sect. \ref{scheme}.  We recall
that a similar procedure was also developed in a semiclassical context
in \cite{Roy}.

A question that we do not address here is that of the behavior of the
part of the spectrum corresponding to the resonant zones of the phase
space. This will be the object of a separate study.

We remark that a normal form theorem with some similarities with the
one presented here was obtained in \cite{PS10} (see Theorem 4.3).
However, as far as we know, our results on the periodic eigenvalues of
the operator \eqref{ope} are new. Furthermore we think that our proof
of the normal form theorem, which is based on symbolic calculus, is
simpler than that of \cite{PS10} and could also have some interest.

\noindent
The paper is organized as follows: in Section \ref{main} we state our
main result, see Theorem \ref{maint}. In Section \ref{scheme} we
explain roughly the strategy of our proof. In Section \ref{pseudo} we
recall some standard facts on the theory of pseudo-differential
operators. In Section \ref{nfs} we state and prove our normal form
result (see Theorem \ref{normalform}) and in Section \ref{prova main
  theorem} we show how Theorem \ref{maint} can be deduced from it. 

\noindent
{\it Acknowledgments:} This research is supported by GNFM. We warmly thank Thomas Kappeler for suggesting some references on the topic and Emanuele Haus, Alberto Maspero and Michela Procesi for many useful discussions and comments.

\section{Main result}\label{main}

First we recall that the dual lattice $\Gamma^*$ is defined by 
\begin{equation}\label{definizione lattice duale}
\Gamma^* := \Big\{ b \in \R^d : b \cdot k \in 2 \pi \Z, \quad \forall
k \in \Gamma \Big\}\, .
\end{equation}

Given any function $u \in L^2(\T^d_\Gamma)$, it can be expanded in
Fourier series as
\begin{equation}\label{Fourier lattice}
u(x) = \sum_{\xi \in \Gamma^*} \widehat u(\xi) e^{\ii x \cdot \xi}\ , \quad \widehat u(\xi) := \frac{1}{\left|\T^d_\Gamma\right|} \int_{\T^d_\Gamma} u(x) e^{- \ii \xi \cdot x}\, d x, \quad \xi \in \Gamma^*\,, 
\end{equation}
where we denoted by $\left|\T^d_\Gamma\right|$ the Lebesgue measure of
$\left|\T^d_\Gamma\right|$. More generally we will denote by
$\left|\cA\right|$ the measure of a measurable set $\cA$. For any $s \geq 0$, we also introduce the Sobolev space $H^s(\T^d_\Gamma)$ defined by 
\begin{equation}\label{definizione spazio sobolev}
H^s(\T^d_\Gamma) := \Big\{ u \in L^2(\T^d_\Gamma) : \| u \|_{H^s} := \Big( \sum_{\xi \in \Gamma^*} \langle \xi \rangle^{2 s} |\widehat u(\xi)|^2 \Big)^{\frac 12} < \infty \Big\}
\end{equation}
where for any $\xi \in \R^d$, we set $\langle \xi \rangle := (1 + |\xi|^2)^{\frac12}$.

\begin{definition}\label{simboli S m delta}
Given $\delta > 0$ and $m \in \R$ we define the symbol class $\sm m$ as
the set of all the functions $a \in \C^\infty(\T^d_\Gamma \times
\R^d)$ such that for any $\alpha, \beta \in \N^d$, there exists
$C_{\alpha,\beta}$ s.t. 
\begin{equation}
  \label{semi}
|\partial_x^\alpha \partial_\xi^\beta a(x, \xi)| \leq C_{\alpha,\beta}
\langle \xi \rangle^{m - \delta |\beta|}\,, \quad \forall (x, \xi) \in
\T^d_\Gamma \times \R^d\,.
\end{equation}
\end{definition}

\noindent
Given a symbol $a \in S^m_\delta$, we define its Fourier coefficients w.r. to the variable $x$ as 
\begin{equation}\label{Fourier transform symbol}
\widehat a(k, \xi) := \frac{1}{|\T^d_\Gamma|} \int_{\T^d_\Gamma} a(x,
\xi) e^{- \ii k \cdot x}dx , \quad \forall (k, \xi) \in \Gamma^*
\times \R^d\,, 
\end{equation}
and the Weyl quantization of a symbol $a \in S^m_\delta$.

\begin{definition}[Weyl quantization]\label{quantizzazione weyl}
Given a symbol $a \in S^m_\delta$, we define its Weyl quantization $ {Op}^w(a)$
as follows: given $u \in C^\infty(\T^d_\Gamma)$, we put
$$
[{Op}^w(a) u](x) = \sum_{k, \xi \in \Gamma^*} \widehat a\Big(k - \xi, \frac{k + \xi}{2} \Big) \widehat u(\xi) e^{\ii k \cdot x}\,. 
$$
Correspondingly, we will say that an operator $A$ is
pseudodifferential of class $\ops m$ if there exists a symbol
$a\in S^m_\delta$ such that $A={Op}^w(a)$
\end{definition}

\begin{definition}
\label{sim.1}
Given a sequence of symbols $\left\{f_j\right\}_{j\geq 0}$ with
$f_j\in S^{m-\rho j}_\delta$ for some $m\in\R$ and $\rho>0$, and a function
$f(x,\xi)$, possibly defined only on $\toro_\Gamma\times \Gamma^*$, we
write 
\begin{equation}
\label{sim.eq}
f\sim \sum_{j}f_j\ ,
\end{equation} 
if for any $N$ there exists $C_N$ s.t.
\begin{equation}
\label{sim.eq.2}
\left|f(x,\xi)-\sum_{j=0}^{N}f_j(x,\xi)\right|\leq{
  C_N}{\langle \xi\rangle^{m-(N+1)\rho}}\ .
\end{equation}
If $f$ is defined only on $\toro_\Gamma\times \Gamma^*$ then
eq. \eqref{sim.eq.2} is valid in such a set.
\end{definition}

The main object of the paper is the spectrum of the operator
\eqref{ope}. { \bf We assume that} there exists $0<\delta<1$ and
$\frak e > 0$ s.t. 
\begin{equation}
  \label{H1}
{\cal V}\in \ops{M - \frak e},\quad {\rm  and}\quad    1 - \frac{\frak
  e}{4} <\delta<1 \ ,
\end{equation}
and that $\cV$ is selfadjoint.

Define
\begin{equation}
  \label{rho}
\rho:= 4 \delta + \frak e - 4>0\ .
\end{equation}
Note that the condition \eqref{H1} on $\delta$ implies that $\rho > 0$. 
 Denote by $B_R(x)$ the open ball of $\R^d$ having radius $R$ and
center $x${, $B_R:=B_R(0)$} and denote by $\sharp E$ the counting measure of a set $E$.

Our main result is the following theorem. 
\begin{theorem}
\label{maint}
Consider the operator
\begin{equation}
  \label{ope1}
H:=[-\Delta]^{M/2}+\cV \ ,
\end{equation}
with $\cV$ fulfilling \eqref{H1}. 
There exists a set $\Omega \subset \R^d$, such that $\Omega \cap
\Gamma^*$ has density one, more precisely one has
\begin{equation}
  \label{density}
1-\frac{\sharp(\Omega\cap\Gamma^* \cap B_R)}{\sharp(B_R \cap
  \Gamma^*)} = \cO(R^{\delta-1})
\end{equation}
and a sequence of symbols $z_j\in S^{M-\fke-j\rho
}$, which depend on $\xi$ only, with the following property: for any
$\xi\in \Omega\cap\Gamma^*$ there exists an eigenvalue $\lambda_\xi$
of \eqref{ope1}
which admits the asymptotic expansion
\begin{equation}
\label{asym}
\lambda_\xi\sim\left|\xi\right|^M+\sum_{j\geq0}z_j(\xi)\ ,\quad \xi\in
\Omega\cap\Gamma^*\ .
\end{equation}
Furthermore, if the symbol $v(x,\xi)$ of $\cV$ is symmetric with
respect to $\xi$, namely
$v(x,\xi)=v(x,-\xi)$, then $\xi\in\Omega$ implies $-\xi\in\Omega$ and
one also has $z_j(\xi)=z_j(-\xi)$, $\forall j$.
\end{theorem}

\begin{remark}
\label{molte}
Actually corresponding to any $\xi\in\Omega\cap\Gamma^*$ we construct
a quasimode $\vphi_\xi$ with the property that for $\xi\not=\xi'$ one
has $\langle\vphi_{\xi},\vphi_{\xi'}\rangle_{L^2}=0$, {
so we construct a number of eigenvalues in
one to one correspondence with $\Omega\cap\Gamma^*$, even in the case
of multiple eigenvalues, which have to be counted with multiplicity.}
\end{remark}

\begin{remark}
  \label{Floq}
Theorem \ref{maint} applies also to the case of Floquet boundary
conditions:
$$
u(x+\gamma)=e^{\im \gamma \cdot \kappa}u(x)\ ,\quad \gamma\in\Gamma\ .
$$
Indeed, the operator \eqref{ope} with Floquet boundary conditions is
unitary equivalent to the operator with symbol
$|\xi-\kappa|^M+v(x,\xi-\kappa)$, which of course fulfills the
assumptions of Theorem \ref{maint}.
\end{remark}

\begin{remark}
\label{split}
In the symmetric case, which in particular
is true for the operator $-\Delta+V(x)$, one has that, for all
$\xi\in\Omega$,
\begin{equation}
\label{split.eq1}
\lambda_\xi-\lambda_{-\xi}=O(|\xi|^{-\infty})\ . 
\end{equation}
\end{remark}

\section{Scheme of the proof}\label{scheme}

The idea of the proof is to perform a ``semiclassical normal form''
(see e.g. \cite{Bam04})
working on the symbol of $\cL$, namely to quantize the classical
normal form procedure for the symbol of $\cL$.

To explain the algorithm we consider the simple
case in which
\begin{equation*}
\cL=-\Delta+V(x)\ ,
\end{equation*}
{and $\Gamma=\Z^d$, namely the lattice generated by the canonical
  basis of $\R^d$. }
In this case $\cL$ is the Weyl quantization of
the classical Hamiltonian
\begin{equation}
  \label{sturm.cl}
h:=|\xi|^2+V(x)\ .
\end{equation}
We are interested in studying the system in the region $\langle \xi
\rangle\gg 1$, in which the potential is a perturbation of the term
$|\xi|^2$. Taking this point of view the perturbative parameter is
$\langle\xi\rangle^{-1}$.   

The classical normal form procedure consists of looking for an
auxiliary Hamiltonian function $g$ s.t. the corresponding time 1 flow
$\phi_g^1$ (namely the time one flow of the corresponding Hamiltonian
system), conjugates $h$ to a new Hamiltonian $h\circ\phi^1_g$ in which
the dependence on the angles $x$ is pushed to higher order. It is well
known that this can be done only in the nonresonant regions of the
action space. The definition of the resonant regions is a key step of
our procedure, hence we are now going to describe it.

By a formal computation one has
$$
h\circ\phi_g=h+\{g,|\xi|^2\}+V+{\rm lower\ order\ terms}\ ,
$$
where $\poisson{.}{.}$ is the Poisson bracket.
The idea is to determine $g$ in such a way that
\begin{equation}
  \label{hom}
\{g,|\xi|^2\}+V= {\rm function \ of} \ \xi\ {\rm only}\ .
\end{equation}
Expanding $g$ and $V$ in Fourier series in $x$, equation \eqref{hom}
takes the form
\begin{equation}
  \label{hom.1}
\im ( k\cdot\xi) \hat g(k,\xi)=\hat V(k,\xi)\ \iff\ \hat
g(k,\xi)=\frac{\hat V(k)}{\im ( k\cdot\xi)}\ ,\quad \forall
k\in\Z^d\setminus\{0 \}\ ,\quad \xi\in\R^d\ ,
\end{equation}
so that the corresponding function $g$ would turn out to be singular
at the dense subset
$$
\bigcup_{k\in\Z^d\setminus\{0 \}}\left\{\xi\in\R^d\ :\ \xi\cdot k=0
\right\}\ .
$$
In classical mechanics it is well known how to solve this problem:
first take advantage of the decay of $|\hat V(k)|$ with
$\left|k\right|$ in order to restrict the union to finite subset of
$\Z^d$, and then remove from the phase space a neighborhood of the so
obtained set.

Since in our case the small parameter is $\langle\xi\rangle^{-1}$, and
we are in a $C^\infty$ context, so that $|\hat V(k, \xi)|$ decays faster
then any inverse power of ${|k|}$, we can proceed as
follows: we fix some $\epsilon>0$ and define
\begin{equation}
  \label{ga}
g(x,\xi):=\sum_{0<|k|<\langle\xi\rangle^\epsilon}\frac{\hat
  V(k,\xi)}{\im k\cdot\xi}\es^{ik\cdot x}\ ,
\end{equation}
but only on the set 
\begin{equation}
  \label{set.0}
\bigcup_{0<|k|<\langle\xi\rangle^\epsilon
}\left\{\xi\in\R^d\ :\ |\xi\cdot k|>\frac{{\mu}}{|k|^{\tau}} \,,
\right\}\ .
\end{equation}
However, even if $g$ is well defined and smooth on the set
\eqref{set.0}, this choice still has a problem: $g$ does not decay
as $\langle\xi\rangle\to\infty$, since the $k$-th term of the sum
\eqref{ga} decays only in the direction $k$. The last remark for the
classical case is that in the domain
$\left|k\cdot\xi\right|\geq C\langle\xi\rangle^\delta$,
with some $\delta>0$, the $k$-th term at r.h.s. of \eqref{ga} decays
as $\langle\xi\rangle^{-\delta}$. This leads to the choice 
$$\mu = \langle \xi \rangle^{\delta}\,$$
in the
formula \eqref{set.0}.
So, we use $g$, but restricted to the domain

\begin{equation}
  \label{set.1111}
\Omega^{(0)}:=\bigcup_{0<|k|<\langle\xi\rangle^\epsilon
}\left\{\xi\in\R^d\ :\ |\xi\cdot
k|>\frac{\langle\xi\rangle^\delta}{|k|^\tau} \right\}\ .
\end{equation}
Furthermore, the measure of the set $\Omega^{(0)}$ is
asymptotically full in the sense that
$$
1-\frac{\left|\Omega^{(0)}\cap B_R\right|}{\left| B_R\right|}= O(R^{\delta-1})\ .
$$

This is the classical procedure that we quantize. As usual in
semiclassical normal form theory, the main remark is that, if $g\in
S^m_\delta$, with $m<\delta$ and $G=\op g$, then $\es^{-\im G}$ is
unitary, the operator $\es^{\im G}\cL\es^{-\im G}$ is pseudodifferential 
and is given by
$$
\es^{\im G}\cL\es^{-\im G}=-\Delta-\im[G,-\Delta]+V+l.h.t.
$$
whose symbol has the form
$$
|\xi|^2+\{g,|\xi|^2\}_M+V+l.h.t.\ ,
$$ where $\{g,|\xi|^2\}_M$ is the so called Moyal bracket, which is the
symbol of the operator $-\im[G,-\Delta]$. Furthermore, since $|\xi|^2$
is quadratic, the Moyal bracket coincides with the Poisson bracket, so
the function $g$ constructed in \eqref{ga} is suitable (after
localization) in order to perform the semiclassical normal form of
$\cL$.

Using $\op g$ in order to transform $\cL$ and iterating the procedure
we conjugate $\cL$ to an operator with symbol
\begin{equation}
  \label{nf.1}
|\xi|^2+z(\xi)+z^{(res)}(x,\xi)+ O(\left|\xi\right|^{-N})\ ,
\end{equation}
with some arbitrarily large $N$. Here $z^{(res)} $ is a symbol localized
in the complement of $\Omega^{(0)}$.

As a last step, we use the equivalence of the Weyl quantization and
the classical quantization in order to show that the operator obtained
by quantizing \eqref{nf.1} acts on $\es^{\im
  x\cdot \xi}$ with $\xi\in\Z^d\cap \Omega^{(0)}$ as a multiplication by
$|\xi|^2+Z(\xi)$ plus an operator which is smoothing of order $N$. Thus
$\es^{\im x\cdot \xi}$ is a quasimode for the quantization of
\eqref{nf.1} and Theorem \ref{main} follows, at least in the case of
Sturm-Liouville type operators. The case of a general torus is
identical to the case just considered and the case where the main
operator is $|\xi|^M$ is easily obtained by just remarking that the
resonant zones of $|\xi|^M$ are the same as those of
$\left|\xi\right|^2$.

\section{Some results on pseudo differential operators}\label{pseudo}

First we give a few lemmas which are quite standard in the framework
of pseudodifferential calculus (see e.g. \cite{robook,Taylor}), and
are here reformulated in a form suitable for our developments.

\begin{lemma}\label{lemma composizione}
Let $m, m' \in \R$, $\delta > 0$, $A = \op a \in OPS^m_\delta$, $B =
\op b \in OPS^{m'}_\delta $. Then  $A B \in OPS^{m +
  m'}_\delta$. Denote by $a\sharp b$ its symbol, then it admits the asymptotic expansion 
$$
\begin{aligned}
& a\sharp b \sim \sum_{j \geq 0} \sigma_j(a,b), \\
& \sigma_j(a,b):= \frac{1}{\ii^j }\sum_{|\alpha| + |\beta| = j} \Big( \frac12 \Big)^{|\alpha|} \Big(- \frac12 \Big)^{|\beta|} (\partial_x^\beta \partial_\xi^\alpha a)( \partial_x^\alpha \partial_\xi^\beta b) \in S^{m + m' - \delta j}_\delta, \quad j \geq 0\,. 
\end{aligned}
$$
Furthermore one has
\begin{equation}
  \label{sym.1}
\sigma_j(a,b)=(-1)^j\sigma_j(b,a)\ .
\end{equation}
If $a$ is symmetric in $\xi$ and $b$ is skewsymmetric in $\xi$, then
$\sigma_j(a,b)$ is symmetric for odd $j$ and skewsymmetric for even $j$. 
\end{lemma}

\begin{corollary}
\label{moyal}
Denoting as usual by $\moyal ab:= \frac{1}{\im}(a\sharp b-b\sharp a) $ the symbol
of $\frac{1}{\im} [A,B]$, one has
$$
\moyal ab\sim\sum_{j\geq 1}\sigma_j^M(a,b)\ , 
$$
with
\begin{equation}
  \label{sym.3}
  \sigma_j^M(a,b)=\left\{
  \begin{matrix}
    0 & \text{if} &j & \text{even}
    \\
    -2  \sigma_j(a,b) & \text{if} &j & \text{odd}
  \end{matrix}\right.\ .
\end{equation}
In particular one has
$$
\moyal ab=  \poisson ab+\sm{m+m'-3\delta}\ , \quad \{ a; b \} :=- \nabla_\xi a \cdot \nabla_x b + \nabla_x a \cdot \nabla_\xi b\,.
$$
Furthermore, if $a$ is symmetric in $\xi$ and $b$ is skewsymmetric in
$\xi$, then $\sigma_j^M$ is symmetric for any $j$. 
\end{corollary}

\begin{lemma}[Adjoint]\label{adjoint}
Let $m \in \R$, $\delta > 0$, $A = \op a \in
OPS^m_\delta$. Then $A^* = \op{\overline a} \in OPS^{m}_\delta$. In
  particular if $a = \overline a$, the operator $\op a$ is self-adjoint. 
\end{lemma}

Given $\eta \in \R$, $\delta > 0$, $g \in \sm{\eta}$, we
define the operator $\adm g :\sm{m}\to \sm{m + \eta -\delta}$ by
$$
		\adm ga:=\moyal ag \ . 
$$ We will also consider its powers $(\adm g)^j,$ which are defined as
                usual. Remark that, for $a\in\sm {m'}$, according to
                Lemma \ref{lemma composizione}, one has $(\adm
                g)^ja\in\sm{m+j(\eta -\delta)}$. Remark that if $g$ is
                skewsymmetric in $\xi$ then $(\adm g)^j$ preserves the
                parity in $\xi$.

Given a selfadjoint pseudodifferential operator $G\in\ops \eta$, we consider
the unitary group generated by $-\im G$, which is denoted, as
usual, by $\es^{-\im\tau G}$, $\tau\in\R$. 

\begin{definition}
\label{conj}
Given a unitary operator $\cU$, we will say that it conjugates an
operator $A$ to an operator $B$ if 
$$
B=\cU A\cU^{-1}\ .
$$
\end{definition}

We recall the following simplified version of Egorov Theorem.

\begin{lemma}\label{teo Egorov}
Fix $\eta \in \R$, and let $g\in\sm{\eta}$, $0 < \delta < 1$be a real valued symbol,
denote $G = \op g \in OPS^\eta$, then $\forall \tau \in [-1, 1]$

\begin{itemize}
			\item[(1)] If $\eta < 1$, then $e^{\ii \tau G} \in {\cal B}\left({H}^{s}; {H}^{s}\right) \quad \forall\ s \geq 0$
			\item[(2)] If $\eta < \delta$, $a\in\sm m$ and $A =\op a$, then
                          $ H := Op^w (h) = e^{\ii \tau G} A e^{- \ii \tau G} \in
                          \ops{m}$ and its symbol admits the
                          asymptotic expansion 
		          \begin{equation}
                            \label{asy.moy}
h \sim \sum_{j \geq 0} \frac{ \tau^j(\adm g)^ja}{j ! }\, .
		          \end{equation}
		          As a consequence the operator $H = Op^w(h) := e^{\ii \tau G} A e^{- \ii \tau G}$ admits the expansion 
		          \begin{equation}\label{espansione egorov astratto}
		          h = a  + \{ a ;  g \}+ S_\delta^{m + 2(\eta - \delta)}\,. 
		          \end{equation}
\item[(3)] If $a$ is even in $\xi$ and $g$ is skewsymmetric in $\xi$ then
  h is even in $\xi$.
 \end{itemize}
\end{lemma}
	
Finally we will need a lemma connecting the classical quantization and
the Weyl quantization. We recall the following definition.

\begin{definition}[Classical quantization]\label{quantizzazione classica}
Given a symbol $a \in S^m_\delta$, we define its classical
quantization $Op^{cl}(a)$ as follows: given $u \in
C^\infty(\T^d_\Gamma)$, we put 
$$
[ {Op}^{cl}(a)u](x) := \sum_{\xi \in \Gamma^*} a(x, \xi) \widehat u(\xi) e^{\ii x \cdot \xi}\,. 
$$ 
\end{definition}
The following lemma is an $\hbar$ independent formulation of Theorem
II-27 of \cite{robook}. 
\begin{lemma}\label{Weyl classic link}
Let $a \in S^m_\delta$. Then there exists a symbol $b \in S^m_\delta$
such that $\op a = {Op}^{cl}(b)$. Furthermore, the symbol $b$ admits the asymptotic expansion 
\begin{equation}\label{asintotica classico weyl}
b\sim \sum_{\alpha \in \N^d} \frac{1}{\ii^{|\alpha| } \alpha! 2^{|\alpha|}} \partial_x^\alpha \partial_\xi^\alpha a\,.
\end{equation}
\end{lemma}
\begin{remark}
In particular, if $a$ is a Fourier multiplier (i.e. independent of
$x$) one has that ${Op}^w(a) = {Op}^{cl}(a)$.  Furthermore we have
that, up to an operator in $\sm{-\infty}$,  supp$(b)\subset$
supp$(a)$.
\end{remark}

\section{The normal form theorem and its proof}\label{nfs}

To state the main result of this section we will use the constant
$\rho$ defined by eq. \eqref{rho}, we fix $\gamma$ s.t. 
{
\begin{equation}\label{raggio minimo lattice duale}
0 < \gamma < r :=  \frac12 \inf_{x \neq y \in \Gamma^*} |x - y| \,, 
\end{equation}
}
and define 
\begin{equation}
  \label{set.1}
\Omega:=\bigcup_{0<|k|<\langle\xi\rangle^\epsilon
}\left\{\xi\in\R^d\ :\ |\xi\cdot
k|>\frac{2\gamma\langle\xi\rangle^\delta}{|k|^\tau} \right\}\ .
\end{equation}
Remark that there exists a neighbourhood of the origin which does not
intersect such a set. 

\begin{theorem}
\label{normalform}
Let $\frak e > 0$ and let $\delta, \rho$ be two constants satisfying \eqref{H1}, \eqref{rho}. Then for any integer $n \geq1 $
there exists a self-adjoint pseudodifferential operator
$\Solh_n = Op^w(g_n)\in\ops{2 - \frak e - \delta -n\rho}$ with symbol $\solh_n$ s.t.
\begin{equation}
  \label{un}
\cU_n:= \es^{\im \Solh_n} \circ \ldots \circ e^{\im \Solh_1}
\end{equation}
conjugates $\cL$ to a pseudodifferential operator $\cL_n$ with symbol $h_n$
of the form
\begin{equation}
  \label{symbolnf}
h_n=\hzero+z^{(n)}+v_n\ , 
\end{equation}
where $v_n\in\sm{M-\fke-n\rho}$ and $z_n\in\sm{M-\fke}$ is such that
$z^{(n)}=\langle z^{(n)}\rangle+z^{(n,res)}$, with $
\langle z^{(n)}\rangle$ independent of $x$ and 
$z^{(n,res)}(x,\xi)=0 $, $\forall \xi\in\Omega$. Furthermore, for any integer $n \geq 0$, there exists a symbol $z_n \in S^{M - \frak e - \rho n}$ such that 
\begin{equation}\label{asintotica z n}
\langle z^{(n)} \rangle (\xi) = \sum_{j = 0}^{n - 1} z_j(\xi)\,.
\end{equation}
Finally, if $h$ is symmetric in $\xi$,
then  the same is true for $v_n$ and $z_n$, whereas $g_1, \ldots, g_n$ are skew symmetric. 
\end{theorem}

The rest of the section is split into few subsections and is devoted to
the proof of this Theorem.

\subsection{Preliminaries and cutoffs}\label{cut}

We start by remarking that, given a symbol $a\in \sm m$, the best constant
$C_{\alpha.\beta}$ for which the inequality \eqref{semi} holds is a
seminorm of $a$. In case it is needed to make reference to the symbol
$a$ we will write $C_{\alpha,\beta}(a)$.  Furthermore, the space $\sm
m$ endowed by such a family of seminorms is a Fr\'echet
space.

Sometimes we will write  $a\lesssim b$ in order to mean there exists a
constant $C$, independent of all the relevant quantities, s.t. $a\leq
Cb$.  
\begin{remark}
\label{decay}
Let $a\in\sm m$ then for any $\alpha \in \N^d$, $N \in \N$
one has that the Fourier coefficients $\widehat a(k,\xi)$ are
estimated by
\begin{equation}\label{marmellata 0}
|\partial_\xi^\alpha \widehat a(k, \xi)| \lesssim \langle k \rangle^{- N} \langle \xi \rangle^{m - \delta |\alpha|}, \quad \forall (k, \xi) \in \Gamma^* \times \R^d\,. 
\end{equation}
Furthermore, defining $\hat a_k(x,\xi):=\widehat a(k, \xi)\es^{\im
  k\cdot x}  $, one has a similar inequality, which implies that $\hat
a_k\in\sm m$ and, furthermore,  for
any $N$ one has 
\begin{equation}
\label{stif}
C_{\alpha,\beta}(\hat a_k)\lesssim\langle
k\rangle^{-N}\ . 
\end{equation}
\end{remark}
We remark that the unwritten constant in the inequalities
\eqref{marmellata 0} and \eqref{stif} depend on
$N,m,\alpha,\beta,\delta$. Of course this dependence is irrelevant for
our developments.

\begin{remark}
\label{prod}
If $a\in\sm m$ and $b\in\sm{m'}$, then $ab\in\sm{m+m'}$ and one has 
\begin{equation}
\label{stism}
C_{\alpha,\beta}(ab)\lesssim\left[\sup_{|\alpha'|+|\beta'|\leq
|\alpha|+|\beta|}C_{\alpha',\beta'}(a)\right] \left[\sup_{|\alpha'|+|\beta'|\leq
|\alpha|+|\beta|}C_{\alpha',\beta'}(b)\right]\ . 
\end{equation}
\end{remark}
Let us consider an even cut-off function $\chi \in C^\infty_c (\R)$
such that ${\rm supp}(\chi) \subseteq [- 2 \gamma, 2 \gamma]$, $0\leq \chi \leq 1$
and $\chi(t) = 1$ for any $t \in [- \gamma, \gamma]$. With its help we define,
for any $k\in\Gamma^*$, 
\begin{equation}\label{cut-off-piccoli-divisori}
\begin{gathered}
\chi_k( \xi) := \chi\Big( \frac{2\left|k\right|^\tau \xi \cdot
  k}{\langle \xi \rangle^{ \delta}} \Big)\ ,
\\
d_k( \xi) := \frac{1}{\xi \cdot k}(1-\chi_k(\xi))\ ,
\\
\tilde\chi_k( \xi) := \chi \left(\frac{|k|}{\langle \csi \rangle^{\varepsilon}}\right)\ . 
\end{gathered}
\end{equation}

\begin{lemma}\label{stima g chi f psi}
The following estimates hold:  
$$
\begin{aligned}
& |\chi_k( \xi)| \lesssim 1\,, \quad |\partial_\xi^\beta \chi_k( \xi)|
  \lesssim \frac{\langle k \rangle^{\tau + |\beta| }}{\langle
    \xi \rangle^{\delta + (|\beta| - 1)}}\,, \qquad \forall \beta \in
  \N^d \setminus \{ 0 \}\,,
  \\
&  |\partial_\xi^\beta d_k( \xi)| \lesssim \frac{\langle k
    \rangle^{\tau + |\beta|}}{\langle \xi \rangle^{\delta( |\beta| +
      1)}}\,, \quad \forall \beta \in \N^d\, ,
  \\
&  |\tilde\chi_k( \xi)| \lesssim 1\,, \quad  |\partial_\xi^\beta \tilde\chi_k( \xi)| \lesssim \frac{\langle k \rangle^{|\beta|}}{\langle \xi \rangle^{\varepsilon + |\beta|}}\,, \quad \forall \beta \in \N^d \backslash \{0\}\,, \\
\end{aligned}
$$
\end{lemma}
\begin{remark}
  \label{cut.esti}
  The above lemma implies that
  \begin{align}
    \label{sti.cut.1}
\chi_k\in \sm0 \ ,\quad {\rm with\ seminorms}\ 
C_{\alpha,\beta}(\chi_k)\lesssim \langle k\rangle^{\tau+|\beta|},  
\\
d_k\in \sm{-\delta} \ ,\quad {\rm with\ seminorms}\ 
C_{\alpha,\beta}(d_k)\lesssim \langle k\rangle^{\tau+|\beta|},  
\\
\tilde\chi_k\in \sm{0} \ ,\quad {\rm with\ seminorms}\ 
C_{\alpha,\beta}(\tilde \chi_k)\lesssim \langle k\rangle^{|\beta|},  
  \end{align}
\end{remark}

\begin{proof}
\noindent
{Estimates of $\chi_k$.} Clearly, by the definition of the cut-off function $\chi$, one has that $\chi_k$ is uniformly bounded by $1$. 
Moreover 
$$
\begin{aligned}
|\partial_\xi \chi_k( \xi)| & \lesssim \Big| \chi'\Big( \frac{\left| k\right|^\tau \xi \cdot k}{\langle \xi \rangle^\delta}  \Big) \langle k \rangle^\tau \Big| \Big| \partial_\xi \Big( \frac{\xi \cdot k}{\langle \xi \rangle^\delta} \Big) \Big| \\
& \lesssim \frac{\langle k \rangle^{\tau + 1}}{\langle \xi \rangle^\delta}
\end{aligned}
$$
The estimates for the higher order derivatives is analogous. 

\bigskip

\noindent
{\sc Estimates of $d_k$}
Note that by the definition \eqref{cut-off-piccoli-divisori}, one has that 
\begin{equation}\label{supporto f psi}
{\rm supp}(d_k) \subseteq \Big\{ (k, \xi) \in \Gamma^* \times \R^d : |\xi \cdot k| \geq \frac{\langle \xi \rangle^\delta}{\langle k \rangle^{\tau}} \Big\}\,. 
\end{equation}
This implies that 
$$
|d_k( \xi)| \lesssim \frac{\langle k \rangle^\tau}{\langle \xi \rangle^\delta}\,. 
$$
Furthermore, one has that 
\begin{align}
|\partial_\xi d_k( \xi)| \lesssim \frac{\langle k \rangle}{|\xi \cdot
  k|^2} \Big|1-\chi_k(\xi) \Big| + \frac{1}{|\xi \cdot k|}
\Big| \chi'_k(\xi) \Big|\,.
\end{align}
By \eqref{supporto f psi}, one has that $\frac{1}{|\xi \cdot k|}
\lesssim \frac{\langle k \rangle^\tau}{\langle \xi \rangle^\delta}$,
for any $(k, \xi) \in {\rm supp}(d_k)$, therefore
$$
\begin{aligned}
|\partial_\xi d_k( \xi)| & \lesssim \frac{\langle k \rangle^{\tau + 1}}{\langle \xi \rangle^{2 \delta}} + \frac{\langle k \rangle^{\tau + 1}}{\langle \xi \rangle^{\delta + 1}} \stackrel{\delta \in (0, 1)}{\lesssim} \frac{\langle k \rangle^{\tau + 1}}{\langle \xi \rangle^{2 \delta}} \,. 
\end{aligned}
$$
The estimate of $\tilde \chi_k$ is similar to the others and is omitted.
\end{proof}
Given $m \in \R$, $\delta > 0$, $a \in S^m_\delta$, we define 
\begin{equation}\label{def a r nr R}
\begin{aligned}
&  \langle a \rangle(\xi) := \frac{1}{|\T^d_\Gamma| } \int_{\T^d_\Gamma} a(x, \xi)\, d x\,, \\
&	a^{(res)}(x, \csi):= \sum_{k \in \Gamma^* \setminus \{ 0 \}}
  \chi_k(\xi) \tilde\chi_k(\xi) \widehat{a}(k, \csi) e^{\ii k \cdot x}\ , \\
  & a^{(nr)}(x, \csi):= \sum_{k \in \Gamma^* \setminus \{ 0 \}}(1-\chi_k(\xi)) \tilde\chi_k(\xi)
\widehat{a}(k, \csi) e^{\ii k \cdot x}\ ,\\
	& a^{(S)}(x, \csi) := \sum_{k \in \Gamma^* \setminus \{ 0 \}}(1-\tilde\chi_k(\xi))\widehat{a}(k, \csi) e^{\ii k \cdot x}\ ,
	\end{aligned}
	\end{equation}
so that one has 
	\begin{equation}\label{splitting simbolo a}
	a = \langle a \rangle + a^{(nr)} + a^{(res)} + a^{(S)}\,.
	\end{equation}
\begin{lemma} \label{lemma simboli cut-off}
 One has that  $\langle a \rangle\,,\, a^{(res)}\,,\, a^{(nr)} \in
 S^{m}_{\delta}$ and $a^{(S)} \in S^{- \infty}$. 
	\end{lemma}
\proof We prove the statement for $a^{(res)}$, all the other estimates
are obtained in the same way.

By the definition of $a^{(res)}$ and recalling the notation introduced in Remark \eqref{decay}, one has 
$a^{(res)}=\sum_{k}\chi_k\tilde\chi_k\hat a_k$. The general term of the
series is in $\sm m$ and, by the estimate \eqref{stif} and Lemma \ref{cut.esti}, its seminorms are estimated
by
$$
C_{\alpha,\beta}(\chi_k\tilde\chi_k\hat a_k
)\lesssim\frac{\langle k\rangle^{\tau+|\beta|}\langle
  k\rangle^{|\beta|}}{\langle
  k\rangle^{N}}\lesssim\frac{1}{\langle
  k\rangle^{N'}}\ ,
$$
with an arbitrary $N'$. Therefore, the series converges in any seminorm and
therefore $a^{(res)}\in \sm m$. Furthermore, its Fourier coefficients
still fulfill \eqref{marmellata 0}.
\qed

\medskip

\noindent
We come to the normal form procedure. First, in order to regularize
the singularity at the origin of the derivatives, we substitute
$|\xi|^M$ with 
\begin{equation}\label{def h 0}
\hzero(\xi)  :=\psi\left(\xi\right)|\xi|^M
\end{equation}
where $\psi \in C^\infty(\R)$ is an even cut-off function such that
$\psi (t) = 0$. for any $t \in [- \gamma, \gamma]$, $0 \leq \psi \leq
1$ and $\psi (t) = 1$ for any $|t| \geq 2\gamma$, where $\gamma$ is
the constant given in \eqref{raggio minimo lattice duale}. Note that
by using such a definition of $\hzero$, for any function $u \in
L^2(\T^d_\Gamma)$, one has that
$$
Op^w(\hzero) u (x) = \sum_{\xi \in \Gamma^*} |\xi|^M \widehat u(\xi) e^{\ii x \cdot \xi}\,. 
$$
\begin{remark}
    \label{restom}
For any $\alpha \in \N^d$, a direct calculation shows that  
  \begin{equation}
    \label{sti.h0}
\left|\partial_{\xi}^{\alpha} \hzero (\xi)\right|\lesssim
\left|\xi\right|^{M-|\alpha|}\ , \quad \forall \xi \in \R^d\,. 
  \end{equation}
\end{remark}

{\subsection{The normal form construction.}}
Then, given $a\in\sm m$, consider
\begin{equation}
  \label{solhomo}
\solh (x,\xi):= \sum_{k\in \Gamma^* \setminus \{ 0 \}}\psi(\xi)\tilde\chi_k(\xi)\left|\xi\right|^{2-M}(-\im)d_k(\xi)
\widehat a(k ,\xi) e^{\ii k \cdot x} \  
\end{equation}
where we recall the definitions given in \eqref{cut-off-piccoli-divisori}

The following Lemma is easily seen to hold

\begin{lemma}
  \label{homole}
  The symbol $\solh$ defined in \eqref{solhomo} belongs to class $\sm{2+m-M-\delta}$ and it satisfies 
\begin{equation}
  \label{homo}
\poisson{\hzero}{\solh
  }+a^{(nr)}\in\sm{-\infty}\,. 
\end{equation}
Moreover, if $a$ is symmetric in $\xi$, then, since
$d_k(\xi)$ is skew symmetric, $\solh$ is skew symmetric too.
\end{lemma}
\proof
Recall the definition of $\hzero$ in \eqref{def h 0}. Using that the cut-off function $\psi$ is a smooth function with compact support, one has that $\nabla_\xi (\psi(\xi) |\xi|^M) = M |\xi|^{M - 2} \xi \psi(\xi ) + S^{- \infty}$, therefore 
\begin{equation}\label{balotelli 0}
 \poisson{h^0}{\solh}(x, \xi) = -  |\xi|^{M - 2} \psi (\xi) \xi \cdot \nabla_x g(x, \xi) + S^{- \infty}_\delta\,.
\end{equation}
In order to solve the equation \eqref{homo}, it is enough to solve 
\begin{equation}\label{equazione omologica astratta 2}
- |\xi|^{M - 2} \psi(\xi) \xi \cdot \nabla_x g  + a^{(n r)} \in  S^{- \infty}_\delta\,. 
\end{equation}
Recalling the definition of $a^{(n r)}$ given in \eqref{def a r nr R}
and the definitions given in \eqref{cut-off-piccoli-divisori}, a
solution of the equation \eqref{equazione omologica astratta 2} is
then given by $g$ defined in \eqref{solhomo}. 
By Lemma \ref{cut.esti}, one gets that the general term in the sum at
r.h. in equation \eqref{solhomo} belongs to $S^{2+m-M-\delta}_\delta$
and for any $\alpha, \beta \in \N$ and $N \in \N$ large enough
(depending on $\alpha, \beta$) its seminorm decay as
$\langle k \rangle^{- N}$, implying that $g \in S^{2 + m - M -
  \delta}_\delta$. Finally, if $a$ is even in $\xi$, using that
$\widetilde \chi_k$ and $\psi$ are even and $d_k$ is odd, one gets
that $g$ is odd in $\xi$ and the proof is concluded.
\qed


\medskip

\noindent{\it Proof of Theorem \ref{normalform}.} 
We describe the induction step of the normal form procedure which allows to Prove the Theorem \ref{normalform}. Assume that $H_n = Op^w(h_n)$ has the form given in \eqref{symbolnf} with $z_n = \langle z_n \rangle + z_n^{(res)} \in S^{M - \frak e}_\delta$ and $v_n \in S^{M - \frak e - n \rho}_\delta$. 
By Lemma \ref{homole} there exists a solution 
\begin{equation}\label{ordine g n + 1}
 g_{n + 1}  \in S^{2 - \frak e - \delta - n\rho}_\delta
 \end{equation}
  of the homological equation 
\begin{equation}\label{equazione omologica iterazione}
 \poisson{\hzero}{ g_{n + 1}} + v_n^{(nr)} \in S^{- \infty}_\delta\,.
\end{equation}
Moreover since the symbol $v_n$ is real valued, then also the symbol $v_n^{(nr)}$ is real valued and therefore $g_{n + 1} = \overline g_{n + 1}$. 
Then, we define $G_{n + 1} := Op^w \big(g_{n + 1} \big)$. Note that $G_{n + 1}$ is self-adjoint since the symbol $g_{n + 1}$ is real valued. Since, by \eqref{H1}, $\delta > 1 - \frac{\frak e}{4} > 1 - \frac{\frak e}{2}$, one obtains that $2 - \frak e - \delta - n \rho < \delta < 1 $. Hence by Lemma \ref{teo Egorov},  $e^{ \ii G_n}, e^{- \ii G_n}$ are well defined linear operators in ${\cal B}(H^s)$ for any $s \geq 0$ and $H_{n + 1} = Op^w(h_{n + 1}) = e^{\im G_{n + 1}} H_n e^{- \im G_{n + 1}} \in OPS^M_\delta$. Furthermore, by applying \eqref{espansione egorov astratto} (with $a = h_n$, $g = g_{n + 1}$, $m = M$, $\eta = 2 - \frak e - \delta - n \rho$), one gets that $h_{n + 1}$ admits the expansion 
$$
\begin{aligned}
h_{n + 1} & = h_n + \poisson{h_n }{g_{n + 1}} + S_\delta^{M + 4 - 2 \frak e - 4 \delta - 2 n \rho  }\,. 
\end{aligned}
$$
Using that, by \eqref{rho}, one gets that $  M + 4 - 2 \frak e - 4 \delta - 2 n \rho \leq M - \frak e - (n + 1) \rho $, implying that $S_\delta^{M + 4 - 2 \frak e - 4 \delta - 2 n \rho } \subseteq S_\delta^{ M - \frak e - (n + 1) \rho}$, hence 
\begin{equation}\label{espansione h n + 1 0}
h_{n + 1}  = h_n +\poisson{h_n }{g_{n + 1}} + S_\delta^{ M - \frak e - (n + 1) \rho}\,. 
\end{equation}

\noindent
Moreover, by \eqref{symbolnf}, using the splitting \eqref{def a r nr R} and Lemma \ref{lemma simboli cut-off}, one has 
\begin{equation}\label{juventus 0}
\begin{aligned}
h_n +\poisson{h_n }{g_{n + 1}} & = \hzero + z_n+ \langle v_n \rangle + v_n^{(res)} + v_n^{(n r)} + \poisson{\hzero }{g_{n + 1}}  \\
& \quad   + \poisson{z_n }{g_{n + 1}}  + \poisson{v_n }{g_{n + 1}}  + S^{- \infty}_\delta\,. 
\end{aligned}
\end{equation}
By \eqref{ordine g n + 1} and using that $z_n \in S_\delta^{M - \frak e}$, $v_n \in S_\delta^{M - \frak e - n \rho} \subseteq S^{M - \frak e}_\delta$.  
$$
\begin{aligned}
& v_n^{(n r)} + \poisson{\hzero }{g_{n + 1}} \in S^{- \infty}_\delta\,, \quad  \poisson{ z_n + v_n }{g_{n + 1}} \in S^{M - \frak e - n \rho - (2 \delta + \frak e - 2) }_\delta\,. \\
\end{aligned}
$$
Note that since $\delta < 1$, one has that $2 \delta + \frak e - 2 > 4 \delta + \frak e - 4 = \rho$ and therefore $\{ g_{n + 1}; z_n + v_n \} \in S^{M - \frak e - (n + 1) \rho}_\delta$. 
hence \eqref{espansione h n + 1 0}, \eqref{juventus 0} imply that 
$$
h_{n + 1} = \hzero + z_{n + 1} + S^{M - \frak e - (n + 1) \rho}
$$
where $z_{n + 1} = \langle z_{n + 1} \rangle + z_{n + 1}^{(res)}$, with $\langle z_{n + 1} \rangle := \langle z_n \rangle + \langle v_n \rangle$ and $z_{n + 1}^{(res)} := z_n^{(res)} + v_n^{(res)}$. The expansions \eqref{symbolnf}, \eqref{asintotica z n} are then proved at the step $n + 1$. Finally, if $v_n$, $z_n$ are symmetric in $\xi$, then Lemma \ref{homole} implies that $g_{n + 1}$ is skew symmetric in $\xi$, hence, by applying Lemma \ref{teo Egorov}, $z_{n + 1}, v_{n + 1}$ are symmetric in $\xi$.

%
%
%
%
%
%
%
%
%
\qed

\subsection{Measure estimates of the non resonant set}\label{measesti}

  In this section we prove that the non resonant set $\Omega$ introduced in Theorem \ref{maint} is of density one. Recall that, according to Definition \eqref{set.1}, 
 	$$
 	\Omega:=\bigcup_{0<|k|<\langle\xi\rangle^\epsilon
 	}\left\{\xi\in\R^d\ :\ |\xi\cdot
 	k|>\frac{2\gamma\langle\xi\rangle^\delta}{|k|^\tau} \right\}\,.
 	$$
 	In particular, we prove the following

\begin{proposition}
  \label{prop.mea}
Assume $\epsilon \leq \frac{\delta}{1 + \tau}$, $\tau>d-1$ and
$R >[r2^{\epsilon(\tau+1)}]^{1/(\delta-\epsilon(\tau+1))}$. Then
one has
\begin{equation}
  \label{stim.meas}
1-\frac{\sharp(\Omega\cap\Gamma^*\cap B_R(0))}{\sharp(\Gamma^*\cap B_R(0))}=\cO(R^{\delta-1})\ .
\end{equation}
\end{proposition}

Given a (measurable) set $\cA$, and a positive parameter $r$ we will
denote
\begin{equation}
  \label{estens}
\cA_r:=\bigcup_{x\in\cA}B_r(x)\ .
\end{equation}

We start by a few remarks that will be useful in order to estimate
the cardinality of $\Omega\cap\Gamma^*$.

\begin{remark}
  \label{m.1}
  There exists a constant $C$ s.t.
  $$
\sharp (\Gamma^*\cap B_R)\geq CR^d\ .
  $$
\end{remark}
\begin{remark}
\label{lemma card misura}
Let $E = \{ x_1, \ldots, x_N \} \subset \R^d$, be a finite subset and
for $r := \inf_{i \neq j} \frac{|x_i - x_j|}{2}$, consider the set
$E_r$ (defined according to \eqref{estens}). Then  
$$
N\equiv \sharp E= \frac{\left|E_r\right|}{\left|B_r(0)\right|}=
\frac{\left|E_r\right|}{\left|B_1(0)\right|r^d}\ .
$$
\end{remark}
{Recall the definition of $r$ as in \eqref{raggio minimo lattice duale};} clearly, one has that for any $\xi_0 \in \Gamma^*$, $B_r(\xi_0) \cap (\Gamma^* \setminus \{ \xi_0 \}) = \emptyset$. 
\begin{remark}
  \label{m.3}
  Given a measurable set, $\cA$, we have
  \begin{equation}
    \label{card.a}
\sharp(\cA\cap\Gamma^*)\leq
\frac{\left|\cA_r\right|}{\left|B_r(0)\right|}\ . 
  \end{equation}
\end{remark}
\begin{remark}
  \label{m.4}
  By the above remark one also has
  \begin{equation}
      \label{stisfer}
\sharp(\Gamma^*\cap B_R(0))\lesssim  R^d \ .
  \end{equation}
\end{remark}

Let 
\begin{equation}\label{Omega R A R}
\Omegar := \Omega \cap B_R\quad \text{and} \quad \aar := B_R \setminus
\Omegar \,;
\end{equation}
In order to estimate the cardinality of $\aar\cap\Gamma^*$ we estimate
the measure of $\aar_r$. To this end we remark that
\begin{equation}
  \label{aar}
\aar_r \subset \bigcup_{0 < |k| < (R+r)^\epsilon} A_{k,r}\,, \quad
A_{k} := \Big\{ \xi \in B_R : |\xi \cdot k| <
\frac{2\gamma R^\delta}{ \left| k\right|^\tau} \Big\}_r \, ,
\end{equation}
{and $A_{k,r}$ is the extension of $A_k$ according to \eqref{estens}.}
In order to estimate the measure of $A_{k,r}$ we will use the
following Lemma.

\begin{lemma}\label{lemma inclusione risonanti}
Assume $\epsilon \leq \frac{\delta}{1 + \tau}$ and
$R >[r2^{\epsilon(\tau+1)}/2\gamma]^{1/(\delta-\epsilon(\tau+1))}$.
Then for any $k \in \Gamma^*$, $0 < |k| < (2R)^\epsilon$, one has that 
\begin{equation}
\label{Atilde}
A_{k,r} \subset \aau := \Big\{ \xi \in B_{2R} : |\xi \cdot k| <
\frac{4\gamma R^\delta}{ \left|k\right|^\tau} \Big\}\,.
\end{equation}
\end{lemma}
\begin{proof}
Let $\xi'\in A_{k,r}$, then there exist $\xi\in A_k$ and $h\in B_r(0)$
s.t. $\xi'=\xi+h$. First one has $\xi'\in B_{2R}(0)$ provided
$r<R$. Furthermore one has 
\begin{equation}
  \label{xih}
|(\xi + h) \cdot k| \leq |\xi \cdot k| + |h \cdot k| \stackrel{\xi \in
  A_k}{<} \frac{2\gamma R^\delta}{\left|k\right|^\tau} + |h| |k| \leq
\frac{2\gamma R^\delta}{\left| k\right|^\tau} + r(2 R)^{\epsilon}\,.
\end{equation}

One has that 
$$
\frac{2\gamma R^\delta}{\left|k\right|^\tau} + r (2R)^{\epsilon} \leq
\frac{4\gamma  R^\delta}{\left|k\right|^\tau}
$$
by taking  $r 2^{\epsilon(\tau+1)}/2\gamma < R^{\delta - \epsilon (1 +
  \tau)}$ which is implied by the assumption. 
\end{proof}

\begin{proposition}
  \label{sti.mea}
Assume $\epsilon \leq \frac{\delta}{1 + \tau}$, $\tau>d-1$ and
$R >[r2^{\epsilon(\tau+1)}/2\gamma]^{1/(\delta-\epsilon(\tau+1))}$. Then
one has
\begin{equation}
  \label{sti.mea.1}
\left|\aar_r\right|\lesssim R^{d+\delta-1}
\end{equation}
\end{proposition}
\proof The proof is standard, we give it here for the sake of
completeness. Since $\aau$ as defined in \eqref{Atilde} is the
intersection of a layer of thickness
$4R^{\delta}/\left|k\right|^\tau|k|$ with a sphere of radius $R$, we
have
$$
\left|\aau\right|\lesssim \frac{R^\delta}{\left|k\right|^{\tau+1}}R^{d-1}\ ,
$$
thus, having fixed some large $R_1$, we have
\begin{align}
  \nonumber
\left|\bigcup_{|k|\leq (2R)^{\epsilon}}\aau\right|\leq
\left|\bigcup_{|k|\in\Gamma^*\setminus\{0\}}\aau\right| \leq
\sum_{|k|\in\Gamma^*\setminus\{0\}}\left|\aau\right|
\\
\nonumber
\lesssim  \sum_{|k|\in\Gamma^*\setminus\{0\}}
\frac{R^\delta}{\left|k\right|^{\tau+1} }R^{d-1} \lesssim R^{\delta+d-1}\sum_{l= 0
}^{+\infty} \sum_{lR_1<|k|\leq(l+1)R_1}\frac{1}{|k|^\tau}
\\
\label{sti.me.4}
\lesssim R^{\delta+d-1}\left( \sum_{0<|k|\leq
  R_1}\frac{1}{\left|k\right|^\tau}+
\sum_{l=1
}^{+\infty}\sum_{lR_1<|k|\leq(l+1)R_1}\frac{1}{(lR_1)^\tau}\right) \ .  
\end{align}
Exploiting again Remark \ref{m.3}, we get, for $l\geq1$,
$$
\sharp[(B_{(l+1)R_1}\setminus B_{lR_1})\cap\Gamma^*]\lesssim
\left|B_{(l+1)R_1}\setminus B_{lR_1}\right|\lesssim l^{d-1}R_1^d\ , 
$$
and thus the bracket at r.h.s. of \eqref{sti.me.4} is bounded by a
constant and the proposition holds.
\qed

We finally show the following 
\begin{proposition}\label{lemma stima misura finale}
For any $R >[r2^{\epsilon(\tau+1)}]^{1/(\delta-\epsilon(\tau+1))}$, one has that 
\begin{equation}\label{maradona - 1}
\frac{\sharp\Big( \Omega \cap B_R \cap \Gamma^* \Big)}{\sharp\Big( B_R \cap \Gamma^* \Big)} = 1 - O(R^{\delta - 1})\,.
\end{equation}
As a consequence, since $0 < \delta < 1$, 
$$
\lim_{R \to + \infty} \frac{\sharp\Big( \Omega \cap B_R \cap \Gamma^* \Big)}{\sharp\Big( B_R \cap \Gamma^* \Big)} = 1\,. 
$$
\end{proposition}
\begin{proof}
By recalling the formula \eqref{Omega R A R}, one has that 
\begin{equation}\label{maradona 0}
\sharp\Big( \Omega \cap B_R \cap \Gamma^* \Big) = \sharp\Big( \Omega^{(R)} \cap \Gamma^* \Big) = \sharp \Big(B_R \cap \Gamma^* \Big) - \sharp \Big( \Omega^{(R, c)}\cap \Gamma^* \Big)\,.
\end{equation}
By Remark \ref{m.3} and Lemma \ref{sti.mea}, one obtains that 
$$
\sharp \Big( \aar\cap \Gamma^* \Big) \lesssim R^{d+\delta-1}
$$
and therefore, using Remark \ref{stisfer} and the formula \eqref{maradona 0} one obtains the claimed estimate \eqref{maradona - 1}. 
\end{proof}
\section{Proof of Theorem \ref{maint}}\label{prova main theorem}
The estimate \eqref{density} follows by Proposition \ref{lemma stima misura finale}.  

We show now that for $\xi\in\Omega$, $e^{\im k\cdot\xi}$ is a
quasimode for $H_n$.
\noindent
By the normal form Theorem \ref{normalform}, for any $n \in \N$, there exists a unitary map ${\cal U}_n \in {\cal B}(H^s)$ for any $s \geq 0$ such that the $H_n = Op^w(h_n) = {\cal U}_n H_0 {\cal U}_n^{- 1}$ satisfies the expansion given in \eqref{symbolnf}, namely 
$$
h_n = h^0 + \langle z_n \rangle + z_n^{(res)}  + v_n 
$$
with $\langle z_n \rangle, z_n^{(res)} \in S^{M - \frak e}_\delta$, $v_n \in S^{M - \frak e - \rho n}_\delta$ and ${\rm supp}(z_n^{(res)}) \subseteq \Omega$. By applying Lemma \ref{Weyl classic link}, one has that there exists $\widetilde z_n^{(res)} \in S^{M - \frak e}_\delta$, $\widetilde v_n \in S^{M - \frak e - \rho n}_\delta$ such that 
\begin{equation}\label{maradona 10}
\begin{aligned}
& Op^w (z_n^{(res)}) = Op^{cl}(\widetilde z_n^{(res)}), \quad Op^w (v_n) = Op^{cl}(\widetilde v_n), \\
& {\rm supp}(\widetilde z_n^{(res)}) = {\rm supp}(z_n^{(res)}) \subseteq \R^d \setminus \Omega, \quad {\rm supp}(v_n ) = {\rm supp}(\widetilde v_n)\,. 
\end{aligned}
\end{equation}
Therefore, given $\xi \in \Omega \cap \Gamma^*$, one gets, by explicit
computation exploiting the definition of $Op^{w}$,
\begin{equation}\label{maradona 11}
\begin{aligned}
Op^w(z_n^{(res)})[e^{\ii x \cdot \xi}] & = Op^{cl}(\widetilde z_n^{(res)})[e^{\ii x \cdot \xi}] = 0 
\end{aligned}
\end{equation}
since $\widetilde z_n^{(res)}(x, \xi) = 0$ for any $\xi \in \Omega$. Moreover, one has that 
\begin{equation}\label{maradona 12}
\begin{aligned}
Op^w(v_n)[e^{\ii x \cdot \xi}] & = Op^{cl}(\widetilde v_n^{(res)})[e^{\ii x \cdot \xi}] =   \widetilde v_n^{(res)}(x, \xi) e^{\ii x \cdot \xi} = O(\langle \xi \rangle^{M - \frak e - \rho n})\,.
\end{aligned}
\end{equation}
Hence \eqref{maradona 10}-\eqref{maradona 12} imply that for any $\xi \in \Omega \cap \Gamma^*$
\begin{equation}\label{bla bla}
H_n [e^{\ii x \cdot \xi}] = \lambda_n(\xi) e^{\ii x \cdot \xi} + O(\langle \xi \rangle^{M - \frak e - \rho n})\,, \quad \lambda_n(\xi) := h^0(\xi) + \langle z_n \rangle(\xi)\,. 
\end{equation}
The existence of one eigenvalue $ O(\langle \xi \rangle^{M - \frak e -
  \rho n})$ close to $ \lambda_n(\xi) $ follows by the standard
quasimode argument. In the case with symmetry, we need to construct
two eigenvalues bifurcating from $ \lambda_n(\xi) $ (note that $
\lambda_n(\xi) $ is even in $\xi$). This situation was studied in
\cite{BKP15}. According to Proposition 5.1, statement (ii) of that
paper, the result
follows from the fact that in such a case
{$\langle e^{\im\xi\cdot x},e^{\im\xi'\cdot x}\rangle_{L^2}=0$}. Of course the same is true
in the case of higher multiplicity. 

We also remark that, defining $\vphi_{n, \xi}:= {\cal U}_n^{-
  1}[e^{\ii x \cdot \xi}]$, it is a quasimode for the original
Hamiltonian. Indeed one has
$$
\begin{aligned}
H_0 \vphi_{n, \xi}  & = {\cal U}_n^{- 1} H_n {\cal U}_n [\vphi_{n,
    \xi} ]=  \lambda_n(\xi) \vphi_{n, \xi} + O(\langle \xi \rangle^{M
  - \frak e - \rho n})\ .
\end{aligned}
$$

\bibliography{biblio}
\bibliographystyle{alpha}

\def\cprime{$'$}

\end{document}